%% file: WSA.tex
\documentclass[conference]{IEEEtran}
\def\BibTeX{{\rm B\kern-.05em{\sc i\kern-.025em b}\kern-.08em
    T\kern-.1667em\lower.7ex\hbox{E}\kern-.125emX}}
\input{preamble}

\usepackage{bbm, dsfont}
\IEEEoverridecommandlockouts
\begin{document}

\title{On the Impact of Oscillator Phase Noise in an IRS-assisted MISO TDD System\\
\thanks{Funded by the Deutsche Forschungsgemeinschaft (DFG, German Research Foundation) under Germany’s Excellence Strategy - EXC 2092 CASA - 390781972.}
}

%\author{
%	\IEEEauthorblockN{Chu Li, Christian Zenger and Aydin Sezgin}
%	\IEEEauthorblockA{Ruhr-Universit\"at Bochum, Germany\\
%		Email:  \{chu.li, aydin.sezgin\}@rub.de}
%	%\thanks{This work is supported by the Federal Ministry of Education and Research (BMBF) of the Federal Republic of Germany (F\"orderkennzeichen 16KIS0658K, SysKit\_HW).}
%}

\author{\IEEEauthorblockN{Chu Li, Aydin Sezgin}
	\IEEEauthorblockA{Ruhr-Universit\"at Bochum, Germany \\
		Email:  \{chu.li, aydin.sezgin\}@rub.de}
	\and
	\IEEEauthorblockN{Zhu Han }
	\IEEEauthorblockA{University of Houston, USA\\
		zhan2@uh.edu}
%	\thanks{C. Li and A. Sezgin are with the
%		Institute of Digital Communication Systems, Ruhr University Bochum,
%		44801 Bochum, Germany (e-mail: chu.li@rub.de;aydin.sezgin@rub.de).}
%	\thanks{M. van Delden is with the Institute of Electronic Circuits, Ruhr University Bochum,
%		44801 Bochum, Germany (e-mail:  marcel.vandelden@est.ruhr-uni-bochum.de).}
%	\thanks{Z. Han is with the Department of Electrical and Computer Engineering in the University of Houston, Houston, TX 77004 USA, and also with the Department of Computer Science and Engineering, Kyung Hee University, Seoul, South Korea, 446-701 (Email: zhan2@uh.edu).
}
%\and
%\and
%\and

\maketitle

	\begin{abstract}
%	Many recent studies have shown that intelligent reflecting surfaces (IRS) can significantly improve the spectral and energy efficiency of wireless communication links. However, most works assume perfect knowledge of channel state information (CSI), which is impractical in real communication systems. In addition, the impact of the hardware imperfections at the transmitter and receiver have not been well studied. In this work, we investigate the impact of the phase noise caused by the local oscillators from both transmitter and receiver. To this end, we propose a linear minimum mean square error (LMMSE) channel estimation algorithm that takes the phase noise into account. Based on the proposed channel estimator, we study the impact of the phase noise on the downlink performance of an IRS-assisted time division duplex (TDD) system. Both analytical and numerical results are presented, where we find that the system with optimized IRS is more robust to the multiplicative phase noise than the system with random IRS. A key observation is that this robustness increases with the number of reflective elements $N$.
	
	Intelligent reflecting surfaces (IRS) have great potential for achieving higher spectral and energy efficiency. However, the expected benefits depend strongly on the accuracy of the channel estimation. Most of the current work assumes perfect channel state information, which is impractical in real communication systems. Moreover, state-of-the-art IRS channel estimation algorithms are proposed under the assumption of perfect transceivers. These algorithms cannot be applied in the case of imperfect transceivers. In this work, we propose a novel channel estimation algorithm that takes into account phase noise from the local oscillator, which is the major contributor to the transceiver impairments. More specifically, we estimate the channel from uplink pilots transmission. Utilizing the obtained channel estimates the downlink ergodic rate is analyzed, where we find that the IRS-assisted system becomes more robust to phase noise as the number of reflective elements increases. Additionally, the impact of additive receiver noise in uplink vanishes when the number of reflective elements approaches infinity.
	% Furthermore, we analyze the asymptotic behavior as the number of reflective elements approaches infinity, which shows that, as $N$ approaches infinity, the transmit power can be scaled down by $\frac{1}{N}$ and $\frac{1}{N^2}$, respectively, for random and optimized IRS, without compromising the received signal to noise ratio (SNR).} 	
\end{abstract}

\section{Introduction}
Intelligent reflection surface (IRS) aided wireless communication systems, as a promising technology to improve the spectral and energy efficiency for 5G and beyond networks, has received increasing attention. An IRS is a thin metal plate consisting of passive scattering elements that can be controlled by a low-cost electronic circuit. Recent works have proven that IRS-assisted systems can achieve higher spectral and energy efficiency at a lower cost than other technologies, such as conventional multi-input multi-output (MIMO) and relay-aided systems \cite{8811733,IRS_relay,7342977}.

However, most prior works assume perfect knowledge of channel state information (CSI), which is highly unlikely given in practice. Especially for IRS-aided systems, obtaining accurate CSI is challenging. Unlike conventional transmitters and receivers that can transmit or receive pilot signals, IRS does not have radio resources or signal processing capabilities to estimate the channel. To address this issue, recent works estimate the cascaded channel instead of estimating the BS to IRS channel and the IRS to user channel separately \cite{9053695} \cite{IRSCE1}. More specifically, in \cite{9053695} the least-square (LS) based channel estimation is proposed, while a minimum mean square mean square error (MMSE) based algorithm is applied in \cite{IRSCE1}. However, these works assume perfect transceivers. Yet, in practical systems phase noise from the local oscillator, which is the main contributor to transceiver impairments, has a detrimental effect on system performance. In particular, high-frequency oscillators suffer from strong phase noise  \cite{voicu2013performance}. Therefore, systems operating in the high-frequency range, such as terahertz (THz), are severely impacted by phase noise \cite{9039743,hillger2020toward}. The impact of additive phase noise in an IRS-assisted system is studied in \cite{9477418}. In addition to the additive phase noise, the transceiver also suffers from multiplicative phase noise. Compared to additive phase noise, it can cause more severe degradation of system performance \cite{668721},\cite{schenk2008rf}. Yet, the impact of the multiplicative phase noise has not been studied in an IRS-assisted system.
%In \cite{IRS_PN2} the spectral and the energy efficiency of an IRS-assisted system are studied considering the hardware impairment at both the base station (BS) and the IRS under perfect CSI.
%Only a few studies have considered the effect of hardware impairment that is detrimental to the system performance. In particular, high-frequency oscillators suffer from strong phase noise. Therefore, systems operating in the high frequency range, such as terahertz (THz), are severely impacted by phase noise.

%More specifically, multiplicative phase noise leads to a cumulative drift of the signal phase, which implies that even in a static scenario, the channel estimate will be time-dependent due to the cumulative phase drift. Such impact has not been considered in an
%IRS-assisted system.
To fill this gap in research, in this work, we study the impact of
the multiplicative phase noise in an IRS-assisted system. We consider an IRS-assisted system with multiplicative phase noise both at the BS and user. A novel channel estimation algorithm is proposed considering the phase noise. Particularly, we assume that the system operates in time division duplex (TDD) mode, the channel estimates are obtained from uplink pilots transmission. Exploiting the channel reciprocity we investigate the system performance in the downlink, more specifically, we derive the ergodic communication rate in closed form. Simulation results verify the correctness of the closed-form expression. We observe that the system becomes more robust against the phase noise as the number of reflective
elements increases. Moreover, the influence of the additive receiver noise in uplink vanishes as the number of reflecting elements grows asymptotically large.

 The rest of this paper is organized as follows. Sec.~\ref{sec:sysmod} describes the system model. In Sec.~\ref{sec:CE_rate} we propose the channel estimation algorithm and analyze the downlink rate. Simulation results are provided in Sec.~\ref{sec:results}. Finally, Sec.~\ref{sec:conclusion} concludes the paper. 

$\mathit{Notation:}$  Boldface lower and upper case symbols are used to denote the vectors and matrices, respectively. ${(\cdot)}^T$, $(\cdot)^*$ and $(\cdot)^H$  represent the transpose, conjugate and Hermitian transpose operators. We use $\mathbb{E}\left[ \cdot\right]  $ to denote the expectation operator, $\diag{(\mathbf{a})}$ is a diagonal matrix with vector $\mathbf{a}$ on its main diagonal, and $\mathbf{I}_N$ ia a $N \times N$ identity matrix. $\tr{(\mathbf{X})}$, $ \left\|  \mathbf{X}\right\| $ and $vec{(\mathbf{X})}$ denote the trace, norm  and vectorization with respect to the matrix $\mathbf{X}$. $\otimes$ is the Kronecker product.

% Notably, this work provides new research directions for IRS-assisted systems, like investigating power scaling laws under imperfect CSI with hardware impairments.

%The
%spectral and the energy efficiency of an IRS-assisted system
%are studied considering both the hardware impairment at the
%BS and at the IRS in [14], in which the hardware impairment is
%modeled as extended error vector magnitude (EEVM). 

%
%\section{Motivation}
%\vspace{-3pt}
\section{System Model}
\label{sec:sysmod}

\begin{figure}
	\centering
	\includegraphics[width=0.8\linewidth]{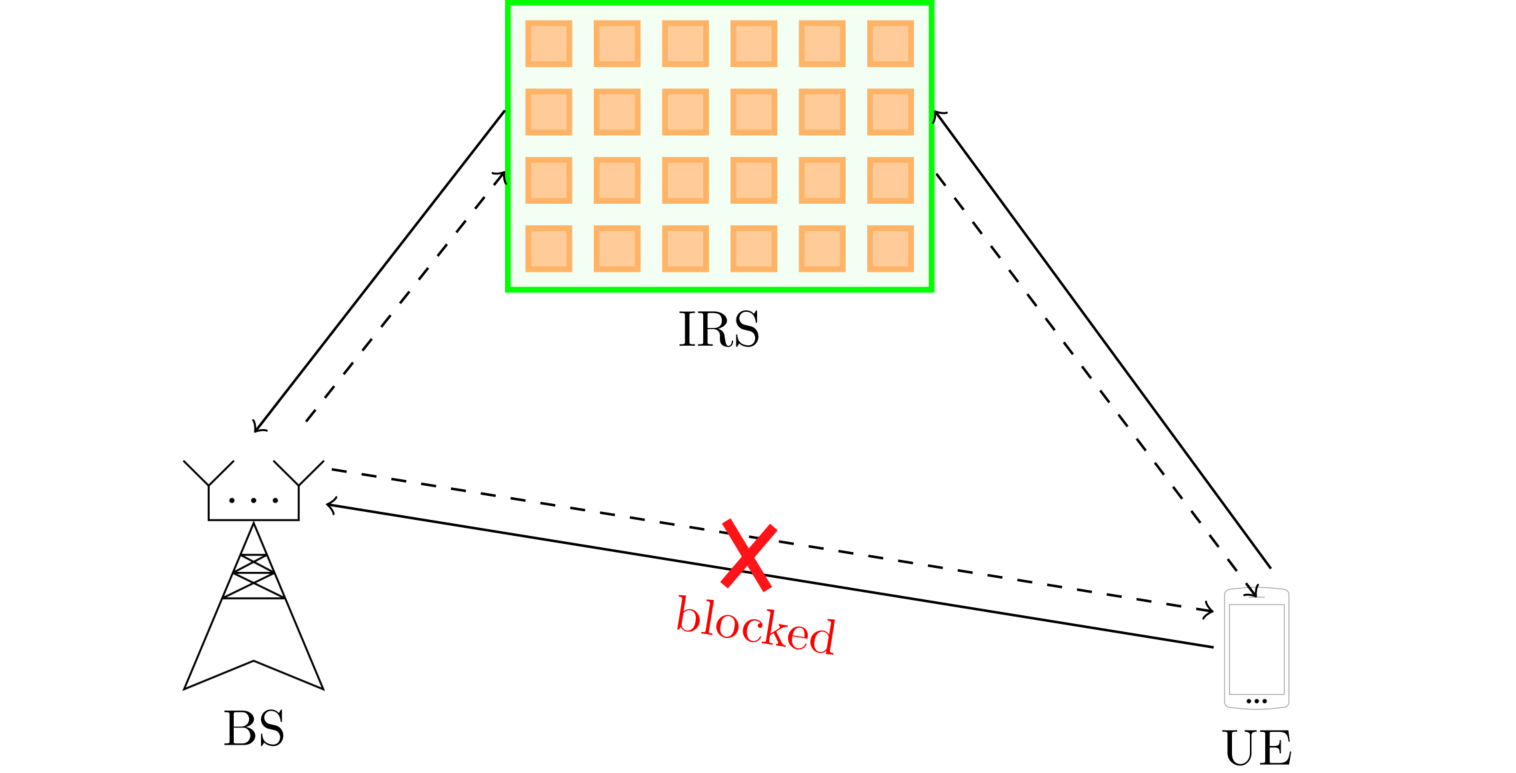}
	\caption{{IRS-assisted single user system: downlink channel (dashed lines) and uplink channel (solid lines)}}
	\label{fig:overview} \vspace{-1.5em}
\end{figure}

We consider a multiple-input single-output (MISO) system in which a BS with $M$ antennas communicates with a user with a single antenna by means of an IRS composed of $N$ reflective elements, as illustrated in \figref{fig:overview}. The phase shift induced by the IRS is controlled by a micro-controller connected to the BS. In addition, we assume that the channel between the BS and the user is blocked by obstacles, which is mostly the case in practice \cite{basar2019wireless}. In this work, a block-fading channel is assumed. We use $T$ to denote the channel coherence time. Furthermore, the communication system operates in TDD mode. We use $[\tau_{1},...,\tau_{B}] \in T$ and $\mathcal{D} \triangleq\{1, \ldots, T\} \backslash\left\{\tau_{1}, \ldots, \tau_{B}\right\}$ to denote the time indices for uplink pilots and downlink data transmission, respectively.
The received signal at the user is modeled as 
%\begin{align}
%\label{eq:dl1}
%y_{t}^{\mathrm{DL}}&= {\left(\mathbf{D}_t\mathbf{h}_{d}+ \mathbf{D}_t\mathbf{H}_{cas}\boldsymbol{\phi}_{t}^{\mathrm{DL}} \right)}^T   \boldsymbol{w}_t \varsigma_t +n_{t}^{\mathrm{DL}},  
%%	\nonumber \\&={\left(\tilde{\mathbf{h}}_{d,t}+ \tilde{\mathbf{H}}_{cas,t}\boldsymbol{\phi}_{t}^{\mathrm{DL}} \right)}^T  \boldsymbol{w}_t  \varsigma_t +n_{t}^{\mathrm{DL}},
%\end{align} 
\begin{align}
\label{eq:dl1}
y_{t}^{\mathrm{dl}}&= {\left(\mathbf{D}_t\mathbf{H}_{1} \diag\left( \boldsymbol{\phi}_{t}^{\mathrm{dl}}\right)  \mathbf{h}_{2} \right)}^T   \boldsymbol{w}_t \varsigma_t +n_{t}^{\mathrm{dl}} \nonumber \\
&= {\left(\mathbf{D}_t\mathbf{H}_{cas}\boldsymbol{\phi}_{t}^{\mathrm{dl}} \right)}^T   \boldsymbol{w}_t \varsigma_t +n_{t}^{\mathrm{dl}},  \quad t \in \mathcal{D},
%	\nonumber \\&={\left(\tilde{\mathbf{h}}_{d,t}+ \tilde{\mathbf{H}}_{cas,t}\boldsymbol{\phi}_{t}^{\mathrm{DL}} \right)}^T  \boldsymbol{w}_t  \varsigma_t +n_{t}^{\mathrm{DL}},
\end{align} 
	where $\mathbf{D}_{t}=\diag{\left(e^{j{\theta}_{t,1}}, e^{j{\theta}_{t,2}},  \cdots,  e^{j{\theta}_{t,M}}  \right)}$ denotes the multiplicative phase-drifts caused by the imperfect local oscillators at BS and user. Also, ${\theta}_{t,m} = {\theta}^{BS}_{t,m} + {\theta}^{UE}_{t} $, where $ {\theta}^{BS}_{t,m}$ and ${\theta}^{UE}_{t}$ are, respectively, the phase-drifts at the $m$-th BS antenna and the user. Furthermore, ${\theta}^{BS}_{t,m} $  and $ {\theta}^{UE}_{t}$ are modeled as a discrete-time independent Wiener process \cite{schenk2008rf}
\begin{align}
\label{eq:phasenoise}
{\theta}^{BS}_{t,m} &= {\theta}^{BS}_{t-1,m}+\Delta  {\theta}^{BS}_{m}, \\	 
{\theta}^{UE}_{t} &= {\theta}^{UE}_{t-1}+ \Delta  {\theta}^{UE},	
\end{align}
where $\Delta  {\theta}^{BS}_{t,m} \sim \mathcal{N}\left(0, \sigma_{BS,m}^{2}\right)$ and $\Delta  {\theta}^{UE}_{t} \sim \mathcal{N}\left(0, \sigma_{UE}^{2}\right)$ denote the random phase increment caused by the imperfect local oscillator at the BS and user, respectively. We assume that each antenna of the BS has its own local oscillator. Besides, we let $\sigma_{BS,m}^{2}=\sigma_{BS}^{2}$ for each oscillator for simplicity. The phase noise variance is given by $\sigma_{BS/UE}^{2} = 4 \pi^2 f_{c}^2 T_s \zeta_{BS/UE}$, where $f_c$, $T_s$, $\zeta_{BS/UE}$ are the carrier frequency, the symbol interval and a constant that depends on the quality of the local oscillator, respectively. 
In addition, $\mathbf{H}_{cas} \in \mathbb{C}^{M \times N}$ represents the IRS-assisted channel given as $\mathbf{H}_{cas} =\mathbf{H}_{1} \diag(\mathbf{h}_{2}) $, where $\mathbf{H}_{1} \in \mathbb{C}^{M \times N} $ and $\mathbf{h}_{2} \in \mathbb{C}^{N}$ denote the BS to IRS and IRS to user channel, respectively. We model the entries of $\mathbf{H}_{cas}$ as independent and identically distributed complex circularly symmetric Gaussian variables with variance ${\beta}_{cas}$.    
% and $\mathbf{h}_d \sim \mathcal{CN}(0, {\beta}_d\mathbf{I}_M)$, $vec(\mathbf{H}_{cas}) \sim \mathcal{CN}(0, {\beta}_1{\beta}_2\mathbf{I}_{MN})$. Furthermore,
%	$\tilde{\mathbf{h}}_{d,t} = \mathbf{D}_t\mathbf{h}_{d} $ represents the direct channel with phase-drift, which corresponds to the first $M$ elements of $\tilde{{\mathbf{h}}}_{t}$ in \eqref{eq:effchannel}, and $\tilde{{\mathbf{H}}}_{cas,t} =\mathbf{D}_t\mathbf{H}_{cas} $ is the cascaded channel with phase-drift corresponding to the remaining $MN$ elements of $\tilde{{\mathbf{h}}}_{t}$. Furthermore, 
Moreover, we use $\boldsymbol{\phi}_{t}^{\mathrm{dl}} \in \mathbb{C}^{N} $ to denote the phase shifts vector at the IRS during the downlink transmission at time $t$, and $\boldsymbol{w}_t \in {\mathbb{C}}^{M} $ is the precoding vector designed according to the channel estimates, which will be introduced in the following section.  $ \varsigma_t \in {\mathbb{C}} $ is the transmit symbol with power constraint $\mathbb{E}\left[  |\varsigma_t| ^2\right]   = P $, and $n_{t}^{\mathrm{DL}} \sim \mathcal{CN}(0, {\sigma}_{\mathrm{d}}^2) $ is the additive complex Gaussian noise in the downlink. Similarly, we model the received signal at the BS as
\begin{align}
\label{eq:up1}
\mathbf{y}_{t}^{\mathrm{ul}}
&= {\left(\mathbf{D}_t\mathbf{H}_{cas}\boldsymbol{\phi}_{t}^{\mathrm{ul}} \right)}  x_t +\mathbf{n}_{t}^{\mathrm{ul}}, \quad t \in [\tau_{1},...,\tau_{B}] , 
%	\nonumber \\&={\left(\tilde{\mathbf{h}}_{d,t}+ \tilde{\mathbf{H}}_{cas,t}\boldsymbol{\phi}_{t}^{\mathrm{DL}} \right)}^T  \boldsymbol{w}_t  \varsigma_t +n_{t}^{\mathrm{DL}},
\end{align} 
where $\boldsymbol{\phi}_{t}^{\mathrm{ul}} \in \mathbb{C}^{N} $ denotes the phase shifts vector at the IRS during the uplink pilots transmission at time $t$, and $x_t$ is the pilots signal with  $\mathbb{E}\left[  |x_t|^2\right]   = 1 $ . Lastly,  $\mathbf{n}_{t}^{\mathrm{ul}} \sim \mathcal{CN}(0, {\sigma}_{\mathrm{u}}^2 \mathbf{I}_{M})  $ is the additive receiver noise at the BS with covariance $ {\sigma}_{\mathrm{u}}^2\mathbf{I}_{M}$.

\section{Channel estimation scheme and downlink analysis}
\label{sec:CE_rate}
In this section, we introduce the MMSE channel estimation algorithm in Sec. \ref{sec:CE}. The obtained channel estimate is utilized to analyze the downlink performance in Sec. \ref{sec:rate}.
 
\subsection{Channel Estimation}
\label{sec:CE}
To study the impact of the phase noise on the system performance, we propose the MMSE based channel estimation algorithm, by which the channel and the phase noise are jointly estimated. Exploiting the channel reciprocity in the TDD mode, we
obtain the downlink channel from the uplink pilots signal. Since the pilot transmission is corrupted by the random phase drifts caused by the imperfect local oscillator, the conventional IRS channel estimation algorithms in \cite{8683663,9053695,IRSCE1} cannot be applied. To solve this problem, we propose an MMSE channel estimator considering the phase noise in the following.

We rewrite the received signal at the BS \eqref{eq:up1} as
\begin{align}
\label{eq:up2}
\mathbf{y}_{t}^{\mathrm{ul}}
&= (  {(\boldsymbol{\phi}_{t}^{ul})}^T  \otimes \mathbf{I}_{M} ) \left(\mathbf{I}_{N} \otimes \mathbf{D}_{t}\right) \mathbf{h} x_t+\mathbf{n}_{t} \nonumber \\
&= (  {(\boldsymbol{\phi}_{t}^{ul})}^T  \otimes \mathbf{I}_{M} ) \tilde{\mathbf{h}}_t x_t+\mathbf{n}_{t}, 
%	\nonumber \\&={\left(\tilde{\mathbf{h}}_{d,t}+ \tilde{\mathbf{H}}_{cas,t}\boldsymbol{\phi}_{t}^{\mathrm{DL}} \right)}^T  \boldsymbol{w}_t  \varsigma_t +n_{t}^{\mathrm{DL}},
\end{align} 
where ${\mathbf{h}} = vec(\mathbf{H}_{cas})$ is a ${MN}$ dimensional vector, and we define $ \tilde{\mathbf{h}} = \left(\mathbf{I}_{N} \otimes \mathbf{D}_{t}\right) \mathbf{h} $ as the effective channel. Note that the real channel $\mathbf{h}$ is constant with the coherence time, while the effective channel $\tilde {\mathbf{h}}$ is, in contrast, time variant due to the random phase noise. During the uplink pilots transmission, we design the IRS phase shifts vector $\boldsymbol{\phi}_{\tau_i}^{ul}$ as the $i$-th row of a $B \times N$ DFT matrix $\boldsymbol{\Phi}$, which is shown to be the optimal design of the IRS during uplink channel estimation \cite{8683663,9053695,IRSCE1}. Next, we introduce the MMSE estimator of $\tilde{\mathbf{h}}_t$.

	\begin{lemma}
	\label{lemma:1}
	Given the combined uplink pilots signal $\boldsymbol{\psi} \triangleq\left[\mathbf{y}^{\mathrm{T}}_{\tau_{1}} \ldots \mathbf{y}^{\mathrm{T}}_{\tau_{B}}\right]^{\mathrm{T}} \in \mathbb{C}^{B M } 
	$, the MMSE channel estimate of the effective channel $\hat{{\mathbf{h}}}_{t}$ is given by 
	\begin{align}
%	\hat{{\mathbf{h}}}_{t} = \left(  ( {\boldsymbol{\Phi}}^H \tilde{\mathbf{D}} {\boldsymbol{\Omega}}^{-1} )  \otimes \mathbf{I}_{M}\right) \boldsymbol{\psi},
	\hat{{\mathbf{h}}}_{t} =  \frac{\beta_{cas}}{N\beta_{cas}+\sigma_{\mathrm{u}}^2} \left( (   {\boldsymbol{\Phi}}^H \tilde{\mathbf{D}} )    \otimes \mathbf{I}_{M}\right) \boldsymbol{\psi},
	\end{align}
	where
	\begin{align}
	\label{eq:Dtilde}
	\tilde{\mathbf{D}}=\diag\left(x^{*}_{{\tau}_1} e^{-\frac{\sigma_{BS}^{2}+\sigma_{UE}^{2}}{2}\left|t-\tau_{1}\right|}  \ldots x^{*}_{{\tau}_B} e^{-\frac{\sigma_{BS}^{2}+\sigma_{UE}^{2}}{2}\left|t-\tau_{B}\right|}  \right).
	\end{align}	
	The corresponding covariance matrix of the channel estimates is given as
	\begin{align}
	\label{eq:covch}
	\mathbf{\Psi}_{t}&=\mathbb{E}\left[\hat{\mathbf{h}}_{t} \hat{\mathbf{h}}_{t}\right] \nonumber \\
	&= \frac{\beta_{cas}^2}{N\beta_{cas}+\sigma_{\mathrm{u}}^2} \left({\boldsymbol{\Phi}}^H \tilde{\mathbf{D}} \tilde{\mathbf{D}}^H {\boldsymbol{\Phi}} \right)\otimes \mathbf{I}_{M},
	\end{align}
	while the estimation error covariance matrix is
	\begin{align}
	\label{eq:coverr}
	\mathbf{C}_{t} &=\mathbb{E}\left[\left(\tilde{\mathbf{h}}_{t}-\hat{\mathbf{h}}_{t}\right)\left(\tilde{\mathbf{h}}_{t}-\hat{\mathbf{h}}_{t}\right)^{\mathrm{H}}\right] \nonumber \\
	&= \beta_{cas} \mathbf{I}_{MN}- \frac{\beta_{cas}^2}{N\beta_{cas}+\sigma_{\mathrm{u}}^2} \left({\boldsymbol{\Phi}}^H \tilde{\mathbf{D}} \tilde{\mathbf{D}}^H {\boldsymbol{\Phi}} \right)\otimes \mathbf{I}_{M}.
	\end{align}
%	and the $(i,j)$-th element of $ {\boldsymbol{\Omega}} \in {\mathbb{C}}^{B \times B}$ is given as	
%	\begin{align}
%	\label{eq:xtilde}
%	\left[{\boldsymbol{\Omega}} \right]_{i, j}=\left\{\begin{array}{ll}  (\beta_d+N\beta_1\beta_2)\left|x_{\tau_{i}}\right|^{2}+\sigma^2, & i=j, \\ (\beta_d-\beta_1\beta_2)x_{\tau_{i}} x_{\tau_{j}}^{*} e^{-\frac{\sigma_{BS}^{2}+\sigma_{UE}^{2}}{2}\left|\tau_{i}-\tau_{j}\right|} , & i \neq j .\end{array}\right .
%	\end{align}
	
\end{lemma}
\begin{proof}
	The proof is provided in Appendix \ref{sec:A1}.
\end{proof}

%\mathbf{G}_{t}\left(\mathbf{I}_{N+1} \otimes \mathbf{D}_{t}\right) \mathbf{h}+\mathbf{n}_{t}
\subsection{Downlink Performance Analysis}
\label{sec:rate}
Utilizing the channel estimates proposed in Sec. \ref{sec:CE}, we investigate the impact of the phase noise on the downlink performance. According to \eqref{eq:dl1}, we define the received instantaneous SNR as
\begin{align}
\label{eq:SNR_ins}
\gamma_{t}=\frac{P}{\sigma_{d}^{2}}\left|\left(\mathbf{D}_t {\mathbf{H}}_{cas} \boldsymbol{\phi}_{t}^{\mathrm{dl}}\right)^{T} \boldsymbol{w}_{t}\right|^{2},
\end{align}	
where the precoding vector $\boldsymbol{w}_{t}$ depends on the channel estimates. More specifically, we design the precoding vector under maximum ratio transmission (MRT), which is given as 
%\begin{align}
%\label{eq:precoding}
%\boldsymbol{w}_{t} =\left(\mathbf{E}\left[\left\|\hat{\mathbf{H}}_{c a s, t} \boldsymbol{\phi}_{t}^{\mathrm{dl}}\right\|^{2}\right]\right)^{-\frac{1}{2}} {\left(\hat{\mathbf{H}}_{c a s, t} \boldsymbol{\phi}_{t}^{\mathrm{dl}} \right) }^*.
%\end{align} 
\begin{align}
\label{eq:precoding}
\boldsymbol{w}_{t} =\frac{{\left(\hat{\mathbf{H}}_{c a s, t} \boldsymbol{\phi}_{t}^{\mathrm{dl}} \right) }^*}{\left(\mathbf{E}\left[\left\|\hat{\mathbf{H}}_{c a s, t} \boldsymbol{\phi}_{t}^{\mathrm{dl}}\right\|^{2}\right]\right)^{\frac{1}{2}}} .
\end{align} 
Here, we normalize the precoding vector over the average of many channel realizations for analytical tractability. Next, we optimize the IRS by maximizing the received SNR given in \eqref{eq:SNR_ins}. Thus, the optimal IRS is given as \cite{8683663,IRSCE1} 
\begin{align}
\label{eq:optIRS}
\boldsymbol{\phi}_{t,opt}^{\mathrm{DL}} =\exp \left(j \angle\left(\tilde{\mathbf{H}}_{cas,t}^{H} \mathbbm{1}_N\right)\right).
\end{align}
Plugging  \eqref{eq:precoding} and \eqref{eq:optIRS} into \eqref{eq:SNR_ins}, the instantaneous SNR can be observed under the optimized IRS. We now define the ergodic communication rate as  
\begin{align}
\label{eq:capacity}
R=\frac{1}{T} \sum_{t \in D} \log _{2}\left(1+\mathbb{E}\left[\gamma_{t}\right]\right).
\end{align}
We notice that the ergodic rate depends on the averaged SNR $\mathbb{E}\left[\gamma_{t}\right]$, which can be naively computed by taking the average over many instantaneous SNR.  However, it leads to high computational complexity, particularly when the communication system has a large number of antennas at the BS or reflective elements at the IRS. To overcome this issue,  we present the averaged SNR in a closed-form in the following theorem. 
\begin{theorem}
	\label{theorem}
	Given the MMSE channel estimates proposed in Sec. \ref{sec:rate} , the averaged SNR with random IRS is given by
	\begin{align}
	\label{eq:thsnr_nodirect}
	&{\bar{\gamma}}_{t,a}
	=	\frac{P}{\sigma_{d}^{2} }\left( \left( M-1\right) N\eta +N\beta_{cas}\right), 
	\end{align}
	while the averaged SNR with optimized IRS is given by
	\begin{align}
	\label{eq:opt}
	{\bar{\gamma}}_{t,opt} =  \frac{P}{\sigma_{d}^{2}} \left( \left((M-1)+ \frac{N \pi}{4}-1\right) N\eta + N\beta_{cas}\right),
	\end{align}
	%	\eqref{eq:opt}. 
%	For $N \gg  \frac{  (M-1) 4}{\pi}$, we have
%	%	\setcounter{equation}{44}
%	\begin{align}
%	\label{eq:thsnr_nodirect22}
%	{\bar{\gamma}}_{t,opt}^{\prime}\approx   \frac{P}{\sigma_{d}^{2}} \left( \left( \frac{N \pi}{4}-1\right) N\eta + N\beta_{cas}\right),
%	\end{align}
	where $\eta = \frac{(\beta_{cas})^2}{N\beta_{cas}+\sigma_{u}^{2}}\sum_{i=1}^{N}e^{-\left( {\sigma_{B S}^{2}+\sigma_{U E}^{2}}\right) (t-i)}$.
\end{theorem}
\begin{proof}
	The proof is provided in Appendix \ref{sec:A2}.
\end{proof}
%
%\begin{corollary}
%The impact of the noise during the channel estimation $\sigma^{2}$ vanishes as $N \rightarrow \infty$. Moreover, as $N$ increases, the system can tolerate stronger phase noise. These statements hold true for both random and optimized IRS.
%\end{corollary}
\noindent It is easy to see that under perfect CSI, i.e., $\sigma_{B S}^{2}=\sigma_{U E}^{2}=0$ and $\sigma_{u}^2 = 0$,  $\eta$ in \eqref{eq:thsnr_nodirect} and \eqref{eq:opt}  is equivalent to the channel gain $\beta_{cas}$.  Thus, we use $\eta$ to denote the gain of the channel estimates.  Furthermore, we have $ \eta  \frac{\text { a.s. }}{N \rightarrow \infty} \frac{\beta_{cas}}{N}
\sum_{i=1}^{N}e^{-\left( {\sigma_{B S}^{2}+\sigma_{U E}^{2}}\right) (t-i)}$, 
{which shows that the impact of the additive noise during channel estimation with variance $\sigma_{u}^{2}$ gradually vanishes as $N \rightarrow \infty $.} Meanwhile, we observe that $\eta$ is increasing with $N$, implying that the system can tolerate stronger phase noise as $N$ increases. The proof is omitted due to lack of space.
%
% provided in in Appendix \ref{sec:A3}.
Finally, by substituting \eqref{eq:thsnr_nodirect} and \eqref{eq:opt} to \eqref{eq:capacity} we observe the ergodic communication rate with random and optimized IRS, respectively.
%
%the distance of $\eta$ between $N$ and $N-1$ is given as
%\begin{align}
%\label{eq:deltaeta}
%\Delta \eta &=    \frac{\beta_{cas}}{N} \sum_{i=1}^{N}e^{-\left( {\sigma_{B S}^{2}+\sigma_{U E}^{2}}\right) (t-i)} - \frac{\beta_{cas}}{N-1} \sum_{i=1}^{N-1}e^{-\left( {\sigma_{B S}^{2}+\sigma_{U E}^{2}}\right) (t-i)} \nonumber \\
%	&\geq 0, 
%\end{align}
%where the quality holds if $\sigma_{B S}^{2}=\sigma_{U E}^{2}=0$. Proof of \eqref{eq:deltaeta} is provided in Appendix. $\Delta \eta$ is greater than zero, implying that the system can tolerate stronger phase noise as $N$ increases.
%Finally, by substituting \eqref{eq:thsnr_nodirect} and \eqref{eq:opt} to \eqref{eq:capacity} we observe the ergodic communication rate with random and optimized IRS, respectively.
%\\ =&\frac{\beta_{cas}e^{-\left( {\sigma_{B S}^{2}+\sigma_{U E}^{2}}\right) t} }{N(N-1)} ((N-1)(\sum_{i=1}^{N-1}e^{-\left( {\sigma_{B S}^{2}+\sigma_{U E}^{2}}\right) i}+e^{(\sigma_{B S}^{2}+\sigma_{U E}^{2})N}))
% Meanwhile, the summation term in $\eta$, i.e., $\sum_{i=1}^{N}e^{-\left( {\sigma_{B S}^{2}+\sigma_{U E}^{2}}\right) (t-i)}$, increases with the number of reflecting elements $N$, which reveals that the system can tolerate stronger phase noise as $N$ increases.

\begin{remark}
	Note that the entries of the channel estimates $\hat{{\mathbf{h}}}_{t}$ in Lemma \ref{lemma:1} are correlated due to the phase noise, since $\mathbf{\Psi}_{t}$ in \eqref{eq:covch} is not a diagonal matrix. However, the variance of the phase noise $\sigma_{BS}^2$ and $\sigma_{UE}^2$ is usually small in practice, which makes the entries on the main diagonal of $\mathbf{\Psi}_{t}$ dominant compared to the other entries. Therefore, for simplicity, we ignore the correlation of the channel estimates in this work. We consider Theorem \ref{theorem} as a reasonable lower bound for the averaged SNR, while an exact analysis considering the correlation will be introduced in our future work.
\end{remark}

%Furthermore, the communication rate at time $t$ is defined as
%\begin{align}
%R_t = \log _{2}\left(1+\gamma_{t}\right).
%\end{align}
\section{Simulation results}
\label{sec:results}
In this section, the system performance in terms of ergodic rate is presented when the proposed channel estimates are applied for downlink data transmission. More specifically, we study how the additive noise during uplink channel estimation and the multiplicative phase noise affect the system performance.

Throughout the simulation, we set the center frequency $f_c = 2.5 \text{ GHz}$, the symbol time interval $T_s = 10^{-7} $ s. Additionally, the number of the BS antennas is set to 16, and the length of the coherence block is $T = 500$. Furthermore, we assume a transmit power $P = 30$ dBm and a noise variance ${\sigma}^2_{\mathrm{d}} = -80$ dBm. The path loss parameter of $\beta_{cas}$ is modeled as 
$
\beta_{cas} = C_0{ \left( \frac{d_{cas}}{D_0}\right) }^{-\alpha},
$
%%Besides, the number of pilots is $N+1$ in the presence of a direct channel $\mathbf{h}_d$ and $N$ when there is no direct channel.
% it is assumed that the increment variance of phase noise is equal for the BS and the user, i.e., $\sigma_{BS}^2=\sigma_{UE}^2 = 4.93 \times 10^{-4} {\text{ rad}}^2$ (if not specified otherwise), which is obtained by setting the center frequency $f_c = 2.5 \text{ GHz}$, the symbol
%time interval $T_s = 10^{-7} $ s and the phase noise parameter $\zeta_{BS} = \zeta_{UE} = 2 \times  10^{-17} {(\text{rad Hz})}^{-1}$.  Furthermore, the length of the coherence block is set to $T=500$.  During the downlink transmission, we assume a transmit power $P = 30$ dBm and a noise variance ${\sigma}^2_{\mathrm{DL}} = -80$ dBm. The distance-dependent path loss parameter of $\mathbf{h}_d$, $\mathbf{H}_1$ and $\mathbf{h}_2$ is modeled as 
%\begin{align}
%\beta (d) = C_0{ \left( \frac{d}{D_0}\right) }^{-\alpha},
%\end{align}
where we set the reference path loss $C_0=-30$ dB, the path distance $d_{cas}=100$ m, and the path loss factor $\alpha$ to 2. The markers in the following figures are the theoretical results according to Theorem \ref{theorem}, while the curves are the simulation results. As described in Remark 1, we also ignore the correlation between the channel estimates in our simulations. We simulate the channel estimate $\hat{{\mathbf{h}}}$ that follows complex Gaussian distribution with covariance $\eta\mathbf{I}_{MN}$. 
% The is obtained by setting the center frequency $f_c = 2.5 \text{ GHz}$, the symbol
%time interval $T_s = 10^{-7} $ s and the phase noise parameter $\zeta_{BS} = \zeta_{UE} = 2 \times  10^{-17} {(\text{rad Hz})}^{-1}$.

% TODO: \usepackage{graphicx} required
\begin{figure*}
	\centering
	\includegraphics[width=0.9\linewidth]{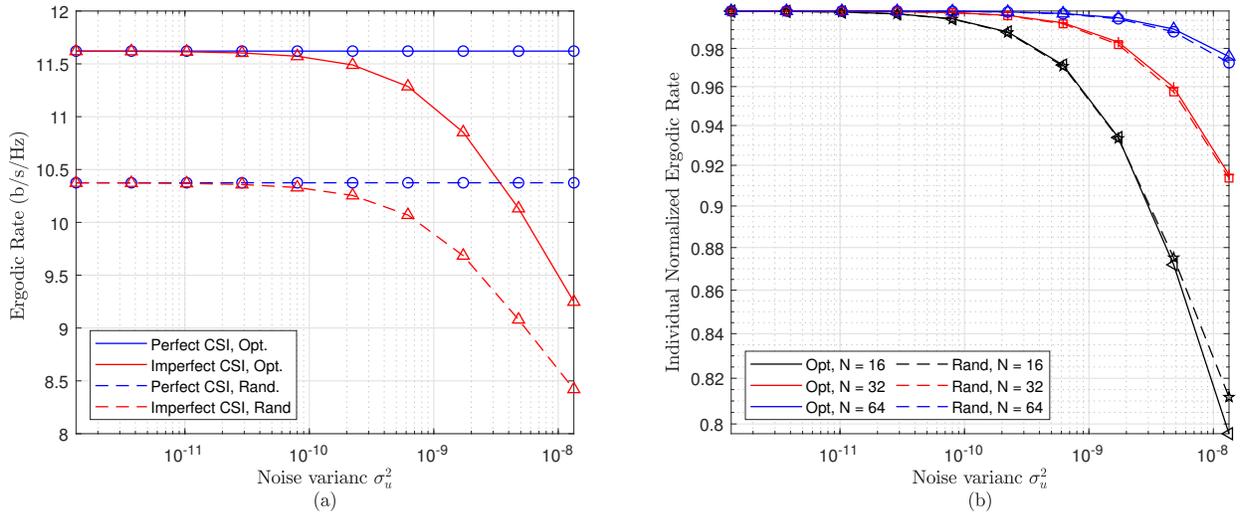}
	\caption{Performance comparison between random and optimized IRS against noise variance $\sigma_u^2$ for $M = 16$, $\zeta_{BS} = \zeta_{UE} = 0$, (a) Ergodic rate with optimized and random IRS for $N=16$, (b) Individual normalized ergodic rate with optimized and random IRS}
	\label{fig:wsafig1}
\end{figure*}

\begin{figure*}
	\centering
	\includegraphics[width=0.9\linewidth]{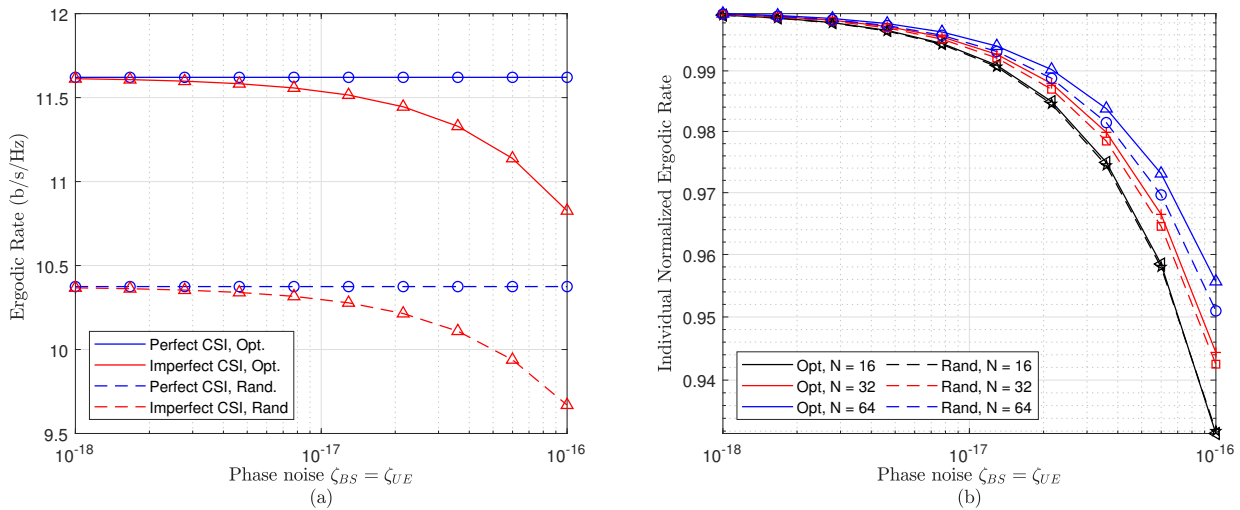}
	\caption{Performance comparison between random and optimized IRS against the phase noise parameter $\zeta_{BS}$ and $ \zeta_{UE}$ for $M = 16$, $\frac{\beta_{cas}}{\sigma^2_u} = 20$ dB, (a) Ergodic rate with optimized and random IRS for $N=16$, (b) Individual normalized ergodic rate with optimized and random IRS}
	\label{fig:wsafig2} \
\end{figure*}

%\begin{figure}
%	%		\centering
%	\begin{subfigure}[b]{0.495\linewidth}
%		\centering
%		\includegraphics[width=\textwidth]{Abstract_fig3}
%		\caption{}
%		\label{fig:figure2a}
%	\end{subfigure}
%	\begin{subfigure}[b]{0.495\linewidth}
%		\centering
%		\includegraphics[width=\textwidth]{Abstract_fig4}
%		\caption{}
%		\label{fig:figure2b}
%	\end{subfigure}
%\vspace{-3ex}
%	\caption{Performance comparison between different number of reflecting elements $N$ against the phase noise $\zeta_{BS} = \zeta_{UE}$ for $M = 8$, $\sigma^2 = 10^{-11}$, (a) Ergodic capacity with perfect and imperfect CSI, (b) Normalized ergodic capacity with imperfect CSI }
%
%	\label{fig:figure2}
%\end{figure}

 We first investigate the impact of the additive noise during the uplink under the assumption of perfect hardware. \figref{fig:wsafig1} (a) plots the ergodic rate with random (Rand.) and optimized (Opt.) IRS utilizing the proposed MMSE channel estimates. Meanwhile, the ergodic rate under the MMSE channel estimates and perfect CSI are compared. We notice that the simulation results agree perfectly with the analytical expressions.  Furthermore, the ergodic rate decreases with increasing $\sigma_{u}^2$. This is because $\sigma_{u}^2$ affects the gain of the channel estimates $\eta$ as introduced in Sec. \ref{sec:rate}.   We also find that the system with optimized IRS performs better than that with random IRS. To study the influence of the number of reflective elements $N$ on system performance, we plot the individual normalized ergodic rate in \figref{fig:wsafig1} (b), which is computed by dividing the ergodic rate under MMSE estimate by it under perfect CSI. It can be clearly seen that the impact of noise becomes weaker as $N$ increases for both random and optimized IRS. An interesting observation is that the system with random IRS is slightly more robust to the additive noise when $N$ is small, i.e. $N=16$, while as $N$ increases, the system with optimized IRS becomes more robust than with random IRS.  
 		
% 	    gap between the perfect and imperfect CSI decreases with increasing $N$, which implies that the impact of noise becomes weaker as $N$ increases. This behavior is clearly shown in \figref{fig:figure1} (b), where the individually normalized ergodic capacity is plotted. }
Next, we study the effects of phase noise on the ergodic communication rate. In \figref{fig:wsafig2} (a), we compare the ergodic rate of the optimized and random IRS with respect to the phase noise parameter. It can be seen that the communication rate under imperfect CSI approaches the rate of the perfect CSI as the phase noise becomes smaller.  Furthermore, we find that the system with optimized IRS still outperforms the system with random IRS when phase noise is taken into account. The individual normalized ergodic rate as a function of the phase noise parameter is shown in \figref{fig:wsafig2} (b), from which we see that the impact of phase noise decreases with increasing $N$. Moreover, with increasing $N$, the robustness of the system with optimized IRS to phase noise becomes higher than that of the system with random IRS.
%{In the final paper, we will introduce the proposed joint phase-channel estimation algorithm, and provide more details on Theorem \ref{theorem}. Furthermore, more simulation results will be given.}

\section{Conclusion}
\label{sec:conclusion} 
In this work, we have investigated the effects of phase noise in an IRS-assisted MISO communication system. We proposed a MMSE-based channel estimation algorithm that takes into account the phase noise caused by the imperfect local oscillator. Using the proposed MMSE estimates, we studied the ergodic downlink rate numerically and analytically, where we find that the robustness to the phase noise increases as the number of reflecting elements increases. We have also shown that the influence of additive noise during uplink channel estimation vanishes as $N$ approaches infinity.
% which trivial that the effect of phase noise can be neglected if the phase noise is small. Moreover, the capacity decreases when the phase noise increases.

\appendix
\section{Appendices}
\subsection{Proof of Lemma 1}
\label{sec:A1}
According to \cite{Kay:2012069} the MMSE estimator of the effective channel is given by
%%	Lemma 1: The LMMSE channel estimate $\hat{{\mathbf{h}}}_{t}$ is given by 
%Given the received pilots signal $\boldsymbol{\psi}$, the LMMSE estimate of the effective channel is given by \cite{Kay:2012069}
%%	\hat{{\mathbf{h}}}_{t}&=\underset{\mathbf{h}_{t}}{\arg \min }\|\mathbf{h}_t- \tilde{\mathbf{h}}_t\|_{2}^{2}, \\ \nonumber
%%&= \mathbb{E}\left\{\tilde {\mathbf{h}}_{t} \boldsymbol{\psi}^{\mathrm{H}}\right\}\left(\mathbb{E}\left\{\boldsymbol{\psi} \boldsymbol{\psi}^{\mathrm{H}}\right\}\right)^{-1} \boldsymbol{\psi}.
%\begin{align}
%\label{eq:ap1}
%\hat{{\mathbf{h}}}_t = \boldsymbol{W}\boldsymbol{\psi},
%\end{align}
%where $\boldsymbol{W}$ is found by minimizing the MSE $\underset{\boldsymbol{W}}{\arg \min }\|\mathbf{h}_t- \tilde{\mathbf{h}}_t\|_{2}^{2}$, and thereby $\boldsymbol{W}$ is given as
\begin{align}
\label{eq:ap1}
\hat{{\mathbf{h}}}_t = \mathbb{E}\left[ \tilde {\mathbf{h}}_{t} \boldsymbol{\psi}^{\mathrm{H}}\right] \left(\mathbb{E}\left[ \boldsymbol{\psi} \boldsymbol{\psi}^{\mathrm{H}}\right] \right)^{-1} \boldsymbol{\psi}.
\end{align}
where the first term is 
\begin{align}
\label{eq:ap3}
& \mathbb{E}\left[ \tilde {\mathbf{h}}_{t} \boldsymbol{\psi}^{\mathrm{H}}\right]  \nonumber \\ 
%& \qquad=\mathbb{E}\left\{  (\mathbf{I}_{N+1} \otimes \mathbf{D}_t )\mathbf{h}  \left( \diag\left[ x_{{\tau}_1} \mathbf{D}_{{\tau}_1},..., x_{{\tau}_B} \mathbf{D}_{{\tau}_B}\right] 
%\left( \boldsymbol{\Phi} \otimes \mathbf{I}_{M} \right) \mathbf{h}\right) ^H \right\} \nonumber\\
\stackrel{(a)}{=}&\mathbb{E}\left[  (\mathbf{I}_{N} \otimes \mathbf{D}_t ) {\mathbf{h}} {\mathbf{h}}^H ({\boldsymbol{\Phi}}^H \otimes \mathbf{I}_{M}) \right.  \nonumber \\ 
& \left. \qquad \qquad \qquad \qquad \diag\left[x^{*}_{{\tau}_1} \mathbf{D}^H_{{\tau}_1},..., x^{*}_{{\tau}_B}, \mathbf{D}^H_{{\tau}_B}\right] \right]  \nonumber \\
%=& \mathbb{E}\left[   {\mathbf{h}} {\mathbf{h}}^H\right]  \mathbb{E}\left[ \right.  (\mathbf{I}_{N} \otimes \mathbf{D}_t ) ({\boldsymbol{\Phi}}^H \otimes \mathbf{I}_{M}) \nonumber \\ 
%& \qquad \qquad \qquad \qquad \diag\left[x^{*}_{{\tau}_1} \mathbf{D}^H_{{\tau}_1},..., x^{*}_{{\tau}_B} \mathbf{D}^H_{{\tau}_B}\right] \left.\right] 
%\nonumber\\
%\stackrel{(a)}{=}& \mathbf{\Lambda} \mathbb{E}\left\{ \right. ({\boldsymbol{\Phi}}^H \otimes \mathbf{I}_{M}) (\mathbf{I}_{N+1} \otimes \mathbf{D}_t )  \nonumber \\ 
%& \qquad \qquad \qquad \qquad \diag\left[x^{*}_{{\tau}_1} \mathbf{D}^H_{{\tau}_1},..., x^{*}_{{\tau}_B} \mathbf{D}^H_{{\tau}_B}\right] \left.\right\} 
%\nonumber\\
=& \beta_{cas}\mathbf{I}_{MN} \mathbb{E}\left\{  ({\boldsymbol{\Phi}}^H \otimes \mathbf{I}_{M})   \right.  \nonumber \\ 
& \left. \qquad  \qquad \qquad \diag\left[x^{*}_{{\tau}_1} \mathbf{D}_t \mathbf{D}^H_{{\tau}_1},..., x^{*}_{{\tau}_B} \mathbf{D}_t \mathbf{D}^H_{{\tau}_B}\right] \right\}
\nonumber\\
%& =\mathbf{\Lambda} \mathbb{E}\left\{ ({\boldsymbol{\Phi}}^H \otimes \mathbf{D}_t) \diag\left[x^{*}_{{\tau}_1} \mathbf{D}^H_{{\tau}_1},..., x^{*}_{{\tau}_B} \mathbf{D}^H_{{\tau}_B}\right]	 \right\} \nonumber\\
\stackrel{(b)}{=}&\beta_{cas}\mathbf{I}_{MN}( {\boldsymbol{\Phi}}^H \otimes \mathbf{I}_{M}) 
\diag\left[ x^{*}_{{\tau}_1} e^{-\frac{\sigma_{BS}^{2}+\sigma_{UE}^{2}}{2}\left|t-\tau_{1}\right|} \mathbf{I}_{M} \right.  \nonumber \\ 
& \left. \qquad \qquad \qquad  \qquad \ldots, x^{*}_{{\tau}_B} e^{-\frac{\sigma_{BS}^{2}+\sigma_{UE}^{2}}{2}\left|t-\tau_{B}\right|} \mathbf{I}_{M}\right] 	\nonumber\\
%=&\mathbf{\Lambda} ( {\boldsymbol{\Phi}}^H \otimes \mathbf{I}_{M}) ( \tilde{\mathbf{D}} \otimes \mathbf{I}_{M}) \nonumber\\
=&(\beta_{cas}{\boldsymbol{\Phi}}^H \tilde{\mathbf{D}})  \otimes \mathbf{I}_{M},
\end{align}	
in which (a) exploits the fact that the additive noise is uncorrelated with the channel $\mathbf{h}$, and (b) utilizes $\mathbb{E} \left[  e^{j \theta_{t1,m}} e^{j \theta_{t2,m}} \right]   = e^{-\frac{\sigma_{BS}^{2}+\sigma_{UE}^{2}}{2}\left|t1-t2\right|} $. $\tilde{\mathbf{D}}$ in \eqref{eq:ap3} is given in \eqref{eq:Dtilde}.
Furthermore, we have
%\begin{align}
%\label{eq:ap4}
%&\mathbb{E}\left\{\boldsymbol{\psi} \boldsymbol{\psi}^{\mathrm{H}}\right\} \nonumber\\
%& =   \diag\left[ x_{{\tau}_1} \mathbf{D}_{{\tau}_1},..., x_{{\tau}_B} \mathbf{D}_{{\tau}_B}\right] 
%\left( \boldsymbol{\Phi} \otimes \mathbf{I}_{M} \right) \mathbf{\Lambda}  \nonumber \\ 
%&({\boldsymbol{\Phi}}^H \otimes \mathbf{I}_{M}) \diag\left[x^{*}_{{\tau}_1} \mathbf{D}^H_{{\tau}_1},..., x^{*}_{{\tau}_B} \mathbf{D}^H_{{\tau}_B}\right]+\sigma^{2}\mathbf{I}_{MB} \nonumber \\
%%& =  \tilde{\mathbf{X}} \otimes \mathbf{I}_{M} + \sigma^2\mathbf{I}_{BM}
%& =  \boldsymbol{\Omega} \otimes \mathbf{I}_{M}
%,\end{align}
%\begin{align}
%\label{eq:ap4}
%&\mathbb{E}\left\{\boldsymbol{\psi} \boldsymbol{\psi}^{\mathrm{H}}\right\} \nonumber\\
%& = \mathbb{E} \left\{\right.  \diag\left[ x_{{\tau}_1} \mathbf{D}_{{\tau}_1},..., x_{{\tau}_B} \mathbf{D}_{{\tau}_B}\right] 
%\left( \boldsymbol{\Phi} \otimes \mathbf{I}_{M} \right) \mathbf{\Lambda}  \nonumber \\ 
%&({\boldsymbol{\Phi}}^H \otimes \mathbf{I}_{M}) \diag\left[x^{*}_{{\tau}_1} \mathbf{D}^H_{{\tau}_1},..., x^{*}_{{\tau}_B} \mathbf{D}^H_{{\tau}_B}\right] \left.\right\} +\sigma^{2}\mathbf{I}_{MB} \nonumber \\
%%& =  \tilde{\mathbf{X}} \otimes \mathbf{I}_{M} + \sigma^2\mathbf{I}_{BM}
%& =  \boldsymbol{\Omega} \otimes \mathbf{I}_{M}
%,\end{align}
\begin{align}
\label{eq:ap4}
&\quad\mathbb{E}\left[ \boldsymbol{\psi} \boldsymbol{\psi}^{\mathrm{H}}\right]  \nonumber\\
& = \mathbb{E} \left\{  \diag\left[ x_{{\tau}_1} \mathbf{D}_{{\tau}_1},..., x_{{\tau}_B} \mathbf{D}_{{\tau}_B}\right] 
\left( \boldsymbol{\Phi} \otimes \mathbf{I}_{M} \right) \mathbf{h} \mathbf{h}^H  \right.  \nonumber \\ 
& \left.({\boldsymbol{\Phi}}^H \otimes \mathbf{I}_{M}) \diag\left[x^{*}_{{\tau}_1} \mathbf{D}^H_{{\tau}_1},..., x^{*}_{{\tau}_B} \mathbf{D}^H_{{\tau}_B}\right] \right\} +\sigma^{2}\mathbf{I}_{MB} \nonumber \\
%& =  \tilde{\mathbf{X}} \otimes \mathbf{I}_{M} + \sigma^2\mathbf{I}_{BM}
%& =  \mathbf{X}  +\sigma^{2}\mathbf{I}_{MB} 
& \stackrel{(a)}{=}  \boldsymbol{\Omega} \otimes \mathbf{I}_{M}
,\end{align}
where (a) is because of the orthogonality of the DFT matrix, and the $(i,j)$-th element of $ \boldsymbol{\Omega}$ is given as
\begin{align}
\label{eq:xtilde}
\left[{\boldsymbol{\Omega}} \right]_{i, j}=\left\{\begin{array}{ll}  N\beta_{cas}\left|x_{\tau_{i}}\right|^{2}+\sigma^2_u, & i=j, \\0 , & i \neq j.\end{array}\right .
\end{align}
%\begin{align}
%	\left[\tilde{\mathbf{X}} \right]_{i, j}=\left\{\begin{array}{ll} \left| x_{{\tau}_i} \right| ^2 (\beta_d+N\beta_1\beta_2), & i=j, \\ x_{\tau_{i}} x_{\tau_{j}}^{*}(\beta_d-\beta_1\beta_2) e^{-\frac{\sigma_{BS}^{2}+\sigma_{UE}^{2}}{2}\left|\tau_{i}-\tau_{j}\right|} , & i \neq j\end{array}\right .
%\end{align}
%and $ \boldsymbol{\Omega}$ is given in \eqref{eq:xtilde}.
For simplicity, we let $\left|x_{\tau_{i}}\right|^{2}=1$. Then, plugging \eqref{eq:ap3} and \eqref{eq:ap4} into \eqref{eq:ap1}, we get the MMSE estimator given as
%\begin{align}
%	\hat{{\mathbf{h}}}_{t} =\mathbf{\Lambda}\left(  ( {\boldsymbol{\Phi}}^H \tilde{\mathbf{D}} )  \otimes \mathbf{I}_{M}\right) {\left( \tilde{\mathbf{X}} \otimes \mathbf{I}_{M} + \sigma^2 \mathbf{I}_{BM}\right) }^{-1}\boldsymbol{\psi}.
%\end{align}
\begin{align}
\hat{{\mathbf{h}}}_{t} &=((\beta_{cas}{\boldsymbol{\Phi}}^H \tilde{\mathbf{D}})  \otimes \mathbf{I}_{M}) {\left( \boldsymbol{\Omega} \otimes \mathbf{I}_{M}\right)  }^{-1}\boldsymbol{\psi} \nonumber \\
& = ((\beta_{cas}{\boldsymbol{\Phi}}^H \tilde{\mathbf{D}})  \otimes \mathbf{I}_{M})  {\left( \boldsymbol{\Omega}^{-1} \otimes \mathbf{I}_{M}\right)  }\boldsymbol{\psi} \nonumber \\
& =  \frac{\beta_{cas}}{N\beta_{cas}+\sigma_{\mathrm{u}}^2} \left( (   {\boldsymbol{\Phi}}^H \tilde{\mathbf{D}} )    \otimes \mathbf{I}_{M}\right) \boldsymbol{\psi}.
\end{align}
Moreover, the covariance matrix of the channel estimate is given as
\begin{align}
\label{eq:Psi}
\mathbf{\Psi}_{t} &= \mathbb{E}\left[ \hat{{\mathbf{h}}}_{t} \hat{{\mathbf{h}}}_{t}^H \right]  \nonumber \\
&=\mathbb{E}\left[ \tilde {\mathbf{h}}_{t} \boldsymbol{\psi}^{\mathrm{H}}\right] \left(\mathbb{E}\left[ \boldsymbol{\psi} \boldsymbol{\psi}^{\mathrm{H}}\right] \right)^{-1} {\left( \mathbb{E}\left[ \tilde {\mathbf{h}}_{t} \boldsymbol{\psi}^{\mathrm{H}}\right] \right)}^H  \nonumber \\
%&=\mathbf{\Lambda}\left(  ( {\boldsymbol{\Phi}}^H \tilde{\mathbf{D}} )  \otimes \mathbf{I}_{M}\right)\nonumber \\
%& \qquad  {(\tilde{\mathbf{X}} \otimes \mathbf{I}_{M} + \sigma^2 \mathbf{I}_{BM})}^{-1} \mathbf{\Lambda}\left(  ( {\boldsymbol{\Phi}} \tilde{\mathbf{D}}^H )  \otimes \mathbf{I}_{M}\right).
&=\beta_{cas}^2\left(  ( {\boldsymbol{\Phi}}^H \tilde{\mathbf{D}} {\boldsymbol{\Omega}}^{-1} )  \otimes \mathbf{I}_{M}\right) \left(  ( \tilde{\mathbf{D}}^H{\boldsymbol{\Phi}}  )  \otimes \mathbf{I}_{M}\right)   \nonumber \\
%&=\mathbf{\Lambda}\left(  ( {\boldsymbol{\Phi}}^H \tilde{\mathbf{D}} {\boldsymbol{\Omega}}^{-1}\tilde{\mathbf{D}}^H{\boldsymbol{\Phi}}  )  \otimes \mathbf{I}_{M}\right)  \mathbf{\Lambda}\nonumber \\
&=\frac{\beta_{cas}^2}{N\beta_{cas}+\sigma_{\mathrm{u}}^2} \left({\boldsymbol{\Phi}}^H \tilde{\mathbf{D}} \tilde{\mathbf{D}}^H {\boldsymbol{\Phi}} \right)\otimes \mathbf{I}_{M}.
\end{align}    
The corresponding estimation error covariance is given as
\begin{align}
\mathbf{C}_t &= \mathbb{E}\left[ \Delta{{\mathbf{h}}}_{t} \Delta{{\mathbf{h}}}_{t}^H \right] \nonumber \\ 
&=\mathbb{E}\left[ \tilde{{\mathbf{h}}}_{t} \tilde{{\mathbf{h}}}_{t}^H \right] -\mathbb{E}\left[ \hat{{\mathbf{h}}}_{t} \hat{{\mathbf{h}}}_{t}^H \right]    \nonumber \\ 
&=\beta_{cas} \mathbf{I}_{MN}- \frac{\beta_{cas}^2}{N\beta_{cas}+\sigma_{\mathrm{u}}^2} \left({\boldsymbol{\Phi}}^H \tilde{\mathbf{D}} \tilde{\mathbf{D}}^H {\boldsymbol{\Phi}} \right)\otimes \mathbf{I}_{M}.
\end{align}
\subsection{Proof of Theorem 1}
\label{sec:A2}
Based on \eqref{eq:SNR_ins} and \eqref{eq:precoding} the averaged SNR is given as
\begin{align}
\label{eq:B1}
\mathbb{E}\left[\hat{\gamma}_{t}\right]&= \frac{P}{\sigma_{d}^{2}}\mathbb{E} \left[ \left|\left(\mathbf{D}_{t}\tilde{ \mathbf{H}}_{c a s} \boldsymbol{\phi}_{t}^{\mathrm{dl}}\right)^{T} \boldsymbol{w}_{t}\right|^{2}\right] \nonumber \\
&=\frac{P}{\sigma_{d}^{2}}\left(\mathbb{E}\left[\left\|\hat{\mathbf{H}}_{c a s, t} \boldsymbol{\phi}_{t}^{\mathrm{dl}}\right\|^{2}\right]\right)^{-1} \nonumber \\
&\quad \qquad  \mathbb{E}\left[\left|  \left(\tilde{\mathbf{H}}_{c a s, t} {\boldsymbol{\phi}}_{t}^{\mathrm{dl}}\right)^T\left(\hat{\mathbf{H}}_{c a s, t} {\boldsymbol{\phi}}_{t}^{\mathrm{dl}}\right)^{*}\right| ^2 \right ],
\end{align}
where the first term of \eqref{eq:B1} for random IRS is simplified as
\begin{align}
\label{eq:B2}
&\mathbb{E}\left[\left\|\hat{\mathbf{H}}_{c a s, t} \boldsymbol{\phi}_{t}^{\mathrm{dl}}\right\|^{2}\right] \nonumber\\ 
= &\mathbb{E}\left[\left\|\left( \left(\phi_{t}^{d l}\right)^{T} \otimes \mathbf{I}_{M}\right) \hat{{\mathbf{h}}}_{t} \right\|^{2}\right]= \mathbb{E}\left[\left\| \hat{{\mathbf{h}}}_{t} \right\|^{2}\right]\nonumber\\
%=& \mathbb{E}\left[\left\| \hat{{\mathbf{h}}}_{t} \right\|^{2}\right] = \tr\left( {\mathbf{\Psi}_{t}}\right) \nonumber\\ 
\stackrel{(a)}{=} & \frac{MN(\beta_{cas})^2}{N\beta_{cas}+\sigma_{u}^{2}}   \sum_{i=1}^{N}e^{-\left( {\sigma_{B S}^{2}+\sigma_{U E}^{2}}\right) \left|t-i\right|} \nonumber\\ 
\stackrel{(b)}{=} & MN\eta,
\end{align}
where (a) follows computing $\tr(\mathbf{\Psi}_{t})$, and (b) is by introducing $\eta =\frac{(\beta_{cas})^2}{N\beta_{cas}+\sigma^{2}}   \sum_{i=1}^{N}e^{-\left( {\sigma_{B S}^{2}+\sigma_{U E}^{2}}\right) \left|t-i\right|} $ for readability. 
Next, we derive $\mathbb{E}\left[\left\|\hat{\mathbf{H}}_{c a s, t} \boldsymbol{\phi}_{t}^{\mathrm{dl}}\right\|^{2}\right]$ for optimized IRS. By simulations we find that averaged SNR with optimized IRS according \eqref{eq:optIRS} is the same with $\boldsymbol{\phi}_{t,opt}^{\mathrm{DL}} =\exp \left(j \angle\left(\hat{\mathbf{h}}_{cas,t}^{H}\right)\right)$, where $\hat{\mathbf{h}}_{cas,t}$ denotes any row of $\tilde{\mathbf{H}}_{cas,t}$. This is also observed in \cite{8811733}. Therefore, we derive the average SNR with optimized IRS according to $\boldsymbol{\phi}_{t,opt}^{\mathrm{DL}} =\exp \left(j \angle\left(\hat{\mathbf{h}}_{cas,t}^{H}\right)\right)$ in the following.   
We use $\hat{\mathbf{H}}_{cas, t}^{\prime}$ to denote a $(M-1) \times N $ submatrix of $\hat{\mathbf{H}}_{cas, t} $ except for the row $\hat{\mathbf{h}}_{{cas}, t}$, then we have
\vspace{-0.5em}
\begin{align}
\label{eq:B2_opt}
\mathbb{E}&\left[\left\|\hat{\mathbf{H}}_{c a s, t} \boldsymbol{\phi}_{t,opt}^{\mathrm{dl}}\right\|^{2}\right] \nonumber\\ 
\stackrel{(a)}{=}&\mathbb{E}\left[\left|\hat{\mathbf{h}}_{\text {cas }, t} \boldsymbol{\phi}_{t, \text { opt }}\right|^{2}+\left\|\hat{\mathbf{H}}_{c a s, t}^{\prime} \boldsymbol {\phi}_{t, o p t}\right\|^{2}\right] \nonumber\\
=&\mathbb{E}\left[\left(\sum_{n=1}^{N}\left|\hat{\mathbf{h}}_{\text {cas }, t}(n)\right|\right)^{2}\right]+\mathbb{E}\left[\left\|\hat{\mathbf{H}}_{c a s, t}^{\prime} \boldsymbol {\phi}_{t, o p t}\right\|^{2}\right] \nonumber\\
\stackrel{(b)}{=}& N^{2} \frac{\pi \eta }{4}+(M-1) N \eta,
\end{align}
where (a) is exploiting the property of the vector norm, and the first term of (b) is because $\tilde{\mathbf{h}}_{{cas,t}}(n)$ follows Rayleigh distribution. While the second term of (b) is derived following same approach as used in \eqref{eq:B2}, since $\boldsymbol {\phi}_{t, o p t}$ is only related to $\hat{\mathbf{h}}_{{cas}, t}$, it can be considered as a random vector in $\mathbb{E}\left[\left\|\hat{\mathbf{H}}_{c a s, t}^{\prime} \boldsymbol {\phi}_{t, o p t}\right\|^{2}\right] $.  Then, the last term of \eqref{eq:B1} with random IRS is obtained as
\begin{align}
\label{eq:B3}
&\mathbb{E}\left[\left|  \left(\tilde{\mathbf{H}}_{c a s, t} \boldsymbol{\phi}_{t}^{\mathrm{dl}}\right)^T\left(\hat{\mathbf{H}}_{c a s, t} \boldsymbol{\phi}_{t}^{\mathrm{dl}}\right)^{*}\right| ^2 \right ] \nonumber \\
= &\mathbb{E}\left[\left|  \left( \left( \hat{\mathbf{H}}_{c a s, t}+ \Delta{\mathbf{H}}_{c a s, t} \right)  {\boldsymbol{\phi}}_{t}\right)^T\left(\hat{\mathbf{H}}_{c a s, t} {\boldsymbol{\phi}}_{t}\right)^{*}\right| ^2 \right ] \nonumber \\
= &\mathbb{E}\left[\left|  \left(  \hat{\mathbf{H}}_{c a s, t}  {\boldsymbol{\phi}}_{t}\right)^T\left(\hat{\mathbf{H}}_{c a s, t} {\boldsymbol{\phi}}_{t}\right)^{*}\right| ^2 \right ] \nonumber \\ 
& +  \mathbb{E}\left[\left|  \left( \Delta{\mathbf{H}}_{c a s, t} {\boldsymbol{\phi}}_{t}\right)^T\left(\hat{\mathbf{H}}_{c a s, t} {\boldsymbol{\phi}}_{t}\right)^{*}\right| ^2 \right ] \nonumber \\
=& {\left( M N {\eta}\right) }^2+M N^2 \beta_{cas} \eta-M(N \eta)^{2}, \vspace{-1em}
\end{align} 
where the last equality follows from algebraic computation of the two independent variables. Similarly, for optimized IRS we have 
%Plugging \eqref{eq:B2} and \eqref{eq:B3} into \eqref{eq:B1}, we observe
%\begin{align}
%\bar{\gamma}_{t, a}=\frac{P}{\sigma_{d}^{2}}\left((M-1) N \eta+N \beta_{cas}\right).
%\end{align}
%Next, we derive the closed-form expression for the averaged SNR with optimized IRS written as
%\begin{align}
%\label{eq:B4}
%\mathbb{E} \left[  {\hat{\gamma}}_{t,opt}^{\prime} \right]  = \frac{\mathbb{E}\left[\left|  \left(\tilde{\mathbf{H}}_{c a s, t} {\boldsymbol{\phi}}_{t,opt}\right)^T\left(\hat{\mathbf{H}}_{c a s, t} {\boldsymbol{\phi}}_{t,opt}\right)^{*}\right| ^2 \right ]}{\mathbb{E}\left[\left\|\hat{\mathbf{H}}_{c a s, t} {\boldsymbol{\phi}}_{t,opt}\right\|^{2}\right]}.
%\end{align}
%Regarding the denominator in \eqref{eq:aveSNR}, we have
%\begin{align}
%\label{eq:B5}
%\mathbb{E}\left[\left\|\hat{\mathbf{H}}_{c a s, t} {\boldsymbol{\phi}}_{t,opt}\right\|^{2}\right] = N^{2} \frac{\pi \eta }{4}+(M-1) N \eta.
%\end{align}
%While the numerator in \eqref{eq:aveSNR} is written as 
%\begin{subequations}
\begin{align}
\label{eq:B3_opt}
&\mathbb{E}\left[\left|  \left(\tilde{\mathbf{H}}_{c a s, t} {\boldsymbol{\phi}}_{t,opt}\right)^T\left(\hat{\mathbf{H}}_{c a s, t} {\boldsymbol{\phi}}_{t,opt}\right)^{*}\right| ^2 \right ] \nonumber \\
%& =  \mathbb{E}\left[\left|  \left(\tilde{\mathbf{h}}_{c a s, t} {\boldsymbol{\phi}}_{t,opt}\right)^T\left(\hat{\mathbf{h}}_{c a s, t} {\boldsymbol{\phi}}_{t,opt}\right)^{*}\right| ^2 \right ]  \\
% &+\mathbb{E}\left[\left|  \left(\tilde{\mathbf{H}}_{c a s, t}^{\prime} {\boldsymbol{\phi}}_{t,opt}\right)^T\left(\hat{\mathbf{H}}_{c a s, t}^{\prime} {\boldsymbol{\phi}}_{t,opt}\right)^{*}\right| ^2 \right ]  \\
% &+ 2 \mathbb{E}\left[  \left(\tilde{\mathbf{h}}_{c a s, t} {\boldsymbol{\phi}}_{t,opt}\right)^H\left(\hat{\mathbf{h}}_{c a s, t} {\boldsymbol{\phi}}_{t,opt}\right)\right ]           \\ 
% & \quad \quad \mathbb{E}\left[ \left(\tilde{\mathbf{H}}_{c a s, t} {\boldsymbol{\phi}}_{t,opt}\right)^T\left(\hat{\mathbf{H}}_{c a s, t} {\boldsymbol{\phi}}_{t,opt}\right)^{*}  \right ]
%
%
%
%
%= &\mathbb{E}\left[\left|  \left( \left( \hat{\mathbf{H}}_{c a s, t}+ \Delta{\mathbf{H}}_{c a s, t} \right)  {\boldsymbol{\phi}}_{t,opt}\right)^T\left(\hat{\mathbf{H}}_{c a s, t} {\boldsymbol{\phi}}_{t,opt}\right)^{*}\right| ^2 \right ] \nonumber \\
%= &\mathbb{E}\left[\left|  \left(  \hat{\mathbf{H}}_{c a s, t}  {\boldsymbol{\phi}}_{t,opt}\right)^T\left(\hat{\mathbf{H}}_{c a s, t} {\boldsymbol{\phi}}_{t,opt}\right)^{*}\right| ^2 \right ] \nonumber \\ 
%& +  \mathbb{E}\left[\left|  \left( \Delta{\mathbf{H}}_{c a s, t} {\boldsymbol{\phi}}_{t,opt}\right)^T\left(\hat{\mathbf{H}}_{c a s, t} {\boldsymbol{\phi}}_{t,opt}\right)^{*}\right| ^2 \right ]\nonumber \\
= &  \left(N^{2} \frac{\pi \eta }{4}+(M-1) N \eta \right)^2\nonumber \\
& +\tr \left( N(\beta_{1}\beta_{2}-\eta)\diag \left( \left[\frac{N^2\pi}{4}\eta, N\eta, ..., N\eta\right] \right) \right)  \nonumber \\
= &  \left(N^{2} \frac{\pi \eta }{4}+(M-1) N \eta \right)^2\nonumber \\
&+  N(\beta_{1}\beta_{2}-\eta)\left(\frac{N^2\pi}{4}\eta+(M-1)N\eta \right).
%  = & \left( N^{2} \frac{\pi \eta }{4}+(M-1) N \eta\right)^2 +  \tr\left( \mathrm{E} \left[\left( \Delta{\mathbf{H}}_{c a s, t} {\boldsymbol{\phi}}_{t,opt}\right) \left( \Delta{\mathbf{H}}_{c a s, t} {\boldsymbol{\phi}}_{t,opt}\right)^H \right]  \right. \nonumber \\
%  &\qquad \left. \mathrm{E}\left[\left(\hat{\mathbf{H}}_{c a s, t} {\boldsymbol{\phi}}_{t,opt}\right)\left(\hat{\mathbf{H}}_{c a s, t} {\boldsymbol{\phi}}_{t,opt}\right)^H\right] \right)  \nonumber \\
% \nonumber 
% \\   
% = & 
\end{align}
%\end{subequations}
Finally, by plugging \eqref{eq:B2} \eqref{eq:B3} into \eqref{eq:B1}, and \eqref{eq:B2_opt} \eqref{eq:B3_opt} into \eqref{eq:B1} we observe theorem 1.  \vspace{-1em}
%
%the average SNR with the optimized IRS given as
%\begin{align}
%&\mathbb{E}\left[{\hat{\gamma}}_{t, o p t}\right] = \frac{P}{\sigma_{d}^{2}}\left(\left((M-1)+\frac{N \pi}{4}-1\right) N \eta+N \beta_{c a s}\right). 
%\end{align} 
\subsection{Proof of $\Delta \eta \geq 0$}
\label{sec:A3}
We observe the distance of $\eta$ between $N$ and $N-1$ as 
\begin{align*}
\Delta \eta = &   \frac{\beta_{cas}}{N} \sum_{i=1}^{N}e^{-\left( {\sigma_{B S}^{2}+\sigma_{U E}^{2}}\right) (t-i)} \nonumber \\ - &\frac{\beta_{cas}}{N-1} \sum_{i=1}^{N-1}e^{-\left( {\sigma_{B S}^{2}+\sigma_{U E}^{2}}\right) (t-i)} \nonumber \\
%=&\frac{\beta_{cas}e^{-\left( {\sigma_{B S}^{2}+\sigma_{U E}^{2}}\right) t}}{N(N-1)}\left( (N-1)(\sum_{i=1}^{N-1}e^{\left( {\sigma_{B S}^{2}+\sigma_{U E}^{2}}\right)i} + e^{\left( {\sigma_{B S}^{2}+\sigma_{U E}^{2}}\right)N}) -N \sum_{i=1}^{N-1}e^{\left( {\sigma_{B S}^{2}+\sigma_{U E}^{2}}\right)i}\right)  \nonumber \\
%=&  \frac{\beta_{cas}e^{-\left( {\sigma_{B S}^{2}+\sigma_{U E}^{2}}\right) t}}{N(N-1)}  \left(  (N-1) e^{\left( {\sigma_{B S}^{2}+\sigma_{U E}^{2}}\right)N} \right.-\nonumber \\& \left.\sum_{i=1}^{N-1}e^{\left( {\sigma_{B S}^{2}+\sigma_{U E}^{2}}\right)i}\right)\nonumber \\
\geq& \frac{\beta_{cas}e^{-\left( {\sigma_{B S}^{2}+\sigma_{U E}^{2}}\right) t}}{N(N-1)}  \left( (N-1) e^{\left( {\sigma_{B S}^{2}+\sigma_{U E}^{2}}\right)N} \right.-\nonumber \\& \left.(N-1)e^{\left( {\sigma_{B S}^{2}+\sigma_{U E}^{2}}\right)(N-1)}\right)\geq  0,
\end{align*}
where the equality holds when $\sigma_{B S}^{2}=\sigma_{U E}^{2}=0$. Thus, $\eta$ is increasing with $N$.
\bibliographystyle{IEEEbib}
\bibliography{refs}
\end{document}

%% file: preamble.tex
\usepackage{cite}
\usepackage{caption}
\usepackage{subcaption}
\usepackage{graphicx}
\usepackage{amsmath}
\usepackage{amssymb}
\usepackage{amsthm}

\usepackage{stfloats}
\usepackage{hyperref}

\usepackage{bm}
\usepackage{algorithmic}
\usepackage{xcolor}
\usepackage{tikz}
\usetikzlibrary{shapes,arrows,positioning,plotmarks,shadows,calc,matrix,fit,patterns}
\usepackage{pgfplots}
\usepackage{todonotes}
\usepackage{datetime}
\usepackage{cancel}
\usepackage{hhline}
\usepackage{cases}
\usepackage{nicefrac}
\usepackage{etoolbox}
\usepackage{algorithm}
\usepackage{algorithmic}
\usepackage{cleveref}
\usetikzlibrary{plotmarks}
\usetikzlibrary{arrows.meta}
\usepackage{grffile}
\usepackage{dsfont}

\newcommand{\figref}[1]{Fig.~\ref{#1}}

\newcommand{\diag}{\text{diag}}

\newcommand{\tr}{\text{tr}}
\newtheoremstyle{remarkmod}
  {\topsep}   % ABOVESPACE
  {\topsep}   % BELOWSPACE
  {\normalfont}  % BODYFONT
  {0pt}       % INDENT (empty value is the same as 0pt)
  {\itshape} % HEADFONT
  {.}         % HEADPUNCT
  {5pt plus 1pt minus 1pt} % HEADSPACE
  {}          % CUSTOM-HEAD-SPEC
\theoremstyle{remarkmod}

\newtheorem{theorem}{Theorem}
\newtheorem{lemma}{Lemma}
\newtheorem{remark}{\textit{Remark}}

%% file: WSA.bbl
\begin{thebibliography}{10}

\bibitem{8811733}
{Q.} Wu and {R.} Zhang,
\newblock ``{Intelligent Reflecting Surface Enhanced Wireless Network via Joint
  Active and Passive Beamforming},''
\newblock {\em IEEE Transactions on Wireless Communications}, vol. 18, no. 11,
  pp. 5394--5409, Nov. 2019.

\bibitem{IRS_relay}
{E.} Björnson, {Ö.} Özdogan, and {E. G.} Larsson,
\newblock ``{Intelligent Reflecting Surface Versus Decode-and-Forward: How
  Large Surfaces are Needed to Beat Relaying?},''
\newblock {\em IEEE Wireless Communications Letters}, vol. 9, no. 2, pp.
  244--248, Feb. 2020.

\bibitem{7342977}
A.~Chaaban and A.~Sezgin,
\newblock ``{Multi-Hop Relaying: An End-to-End Delay Analysis},''
\newblock {\em IEEE Transactions on Wireless Communications}, vol. 15, no. 4,
  pp. 2552--2561, 2016.

\bibitem{9053695}
T.~L. {Jensen} and E.~{De Carvalho},
\newblock ``{An Optimal Channel Estimation Scheme for Intelligent Reflecting
  Surfaces Based on a Minimum Variance Unbiased Estimator},''
\newblock in {\em IEEE International Conference on Acoustics, Speech and Signal
  Processing (ICASSP)}, May 2020, pp. 5000--5004.

\bibitem{IRSCE1}
Q.~U.~A. {Nadeem}, H.~{Alwazani}, A.~{Kammoun}, A.~{Chaaban}, M.~{Debbah}, and
  M.~S. {Alouini},
\newblock ``{Intelligent Reflecting Surface-Assisted Multi-User MISO
  Communication: Channel Estimation and Beamforming Design},''
\newblock {\em IEEE Open Journal of the Communications Society}, vol. 1, pp.
  661--680, May 2020.

\bibitem{voicu2013performance}
D.~Zito M.~Voicu, D.~Pepe,
\newblock ``{Performance and trends in millimetre-wave CMOS oscillators for
  emerging wireless applications},''
\newblock {\em International Journal of Microwave Science and Technology},
  2013.

\bibitem{9039743}
A.~A.~Boulogeorgos E.~N.~Papasotiriou and A.~Alexiou,
\newblock ``{Performance Analysis of THz Wireless Systems in the Presence of
  Antenna Misalignment and Phase Noise},''
\newblock {\em IEEE Communications Letters}, vol. 24, no. 6, pp. 1211--1215,
  Jun. 2020.

\bibitem{hillger2020toward}
P.~Hillger, M.~van Delden, U.~S.~M. Thanthrige, A.~M. Ahmed, J.~Wittemeier,
  K.~Arzi, et~al.,
\newblock ``Toward mobile integrated electronic systems at {THz} frequencies,''
\newblock {\em Journal of Infrared, Millimeter, and Terahertz Waves}, vol. 41,
  no. 7, pp. 846--869, 2020.

\bibitem{9477418}
{A.} Papazafeiropoulos, {C.} Pan, A.~Elbir, V.~Nguyen, {P.} Kourtessis, and
  {S.} Chatzinotas,
\newblock ``Asymptotic analysis of max-min weighted sinr for irs-assisted miso
  systems with hardware impairments,''
\newblock {\em IEEE Wireless Communications Letters}, pp. 1--1, 2021.

\bibitem{668721}
L.~Tomba,
\newblock ``{On the effect of Wiener phase noise in OFDM systems},''
\newblock {\em IEEE Transactions on Communications}, vol. 46, no. 5, pp.
  580--583, May 1998.

\bibitem{schenk2008rf}
T.~Schenk,
\newblock {\em {RF imperfections in high-rate wireless systems: impact and
  digital compensation}},
\newblock Springer Science \& Business Media, 2008.

\bibitem{basar2019wireless}
E.~Basar, M.~Di~Renzo, J.~De~Rosny, M.~Debbah, M.~Alouini, and R.~Zhang,
\newblock ``{Wireless communications through reconfigurable intelligent
  surfaces},''
\newblock {\em IEEE access}, vol. 7, pp. 116753--116773, 2019.

\bibitem{8683663}
D.~{Mishra} and H.~{Johansson},
\newblock ``{Channel Estimation and Low-complexity Beamforming Design for
  Passive Intelligent Surface Assisted MISO Wireless Energy Transfer},''
\newblock in {\em IEEE International Conference on Acoustics, Speech and Signal
  Processing (ICASSP)}, May 2019, pp. 4659--4663.

\bibitem{Kay:2012069}
S.~M. Kay,
\newblock {\em {Fundamentals of statistical signal processing}},
\newblock Prentice-Hall signal processing series. Prentice Hall PTR, Upper
  Saddle River, NJ, 1993.

\end{thebibliography}
